\documentclass[11pt, reqno]{article}
\pdfoutput=1 

\usepackage{jheppub}
\usepackage{amssymb}
\usepackage{amsmath}
\usepackage[usenames,dvipsnames]{xcolor}
\usepackage{epsfig}
\usepackage{dcolumn}
\usepackage{tikz}
\usetikzlibrary{shapes.geometric, arrows}
\usepackage{upgreek}
\usepackage{setspace}
\usepackage{enumitem}
\usepackage{array,multirow,bigdelim,arydshln}
\usepackage{appendix}
\usepackage{xparse}
\usepackage[utf8]{inputenc}
\hypersetup{
	colorlinks,
	urlcolor=Maroon,
	linkcolor=Maroon,
	citecolor=Maroon
}

\usepackage{amsthm}
\newtheorem{thm}{Theorem}[section]

\newtheorem{problem}[thm]{Problem}
\newtheorem{lem}[thm]{Lemma}

\newtheorem{defn}[thm]{Definition}
\theoremstyle{definition}

\usepackage{mathtools}

\NewDocumentCommand{\binomial}{omm}
 {%
  \genfrac(){0pt}{}{#2}{#3}%
  \IfValueT{#1}{_{\!#1}}%
 }
\NewDocumentCommand{\eulerian}{omm}
 {%
  \genfrac<>{0pt}{}{#2}{#3}%
  \IfValueT{#1}{_{\!#1}}%
 }

\def \s {\sigma}

\usepackage{latexsym}
\usepackage{tikz}

\title{Compatible Cycles and CHY Integrals}

\author[a]{Freddy Cachazo,}\emailAdd{fcachazo@pitp.ca}
\author[b]{Karen Yeats,}\emailAdd{kayeats@uwaterloo.ca}
\author[b]{and Samuel Yusim}\emailAdd{syusim@edu.uwaterloo.ca}

\affiliation[a]{Perimeter Institute for Theoretical Physics, Waterloo, ON N2L 2Y5, Canada}
\affiliation[b]{Department of Combinatorics $\&$ Optimization, University of Waterloo, Waterloo, ON N2L 3G1, Canada}

\abstract{The CHY construction naturally associates a vector in $\mathbb{R}^{(n-3)!}$ to every 2-regular graph with $n$ vertices. Partial amplitudes in the biadjoint scalar theory are given by the inner product of vectors associated with a pair of cycles. In this work we study the problem of extending the computation to pairs of arbitrary 2-regular graphs. This requires the construction of {\it compatible} cycles, i.e. cycles such that their union with a 2-regular graph admits a Hamiltonian decomposition. We prove that there are at least $(n-2)!/4$ such cycles for any 2-regular graph. We also find a connection to breakpoint graphs when the initial 2-regular graph only has double edges. We end with a comparison of the lower bound on the number of randomly selected cycles needed to generate a basis of $\mathbb{R}^{(n-3)!}$, using the super Catalan numbers and our lower bound for compatible cycles.}

\begin{document}
\maketitle
\addtocontents{toc}{\protect\setcounter{tocdepth}{1}}
\def \tr {\nonumber\\}
\def \la  {\langle}
\def \ra {\rangle}
\def\hset{\texttt{h}}
\def\gset{\texttt{g}}
\def\sset{\texttt{s}}
\def \be {\begin{equation}}
\def \ee {\end{equation}}
\def \ba {\begin{eqnarray}}
\def \ea {\end{eqnarray}}
\def \k {\kappa}
\def \h {\hbar}
\def \r {\rho}
\def \l {\lambda}
\def \be {\begin{equation}}
\def \en {\end{equation}}
\def \bes {\begin{eqnarray}}
\def \ens {\end{eqnarray}}
\def \red {\color{Maroon}}
\def \pt {{\rm PT}}
\def \s {\textsf{s}}
\def \t {\textsf{t}}
\def \ls {{\rm LS}}
\def \ma {\Upsilon}
\def \ort {\textsf{O}}

\numberwithin{equation}{section}

\section{Introduction}

Scattering amplitudes of massless particles are very constrained by physical requirements such as locality and unitarity (see e.g. \cite{Elvang:2013cua,Benincasa:2007xk}). In 2013, He, Yuan and one of the authors, introduced the CHY formalism which encodes locality and unitarity into the structure of the moduli space of punctured Riemann spheres \cite{Cachazo:2013gna,Cachazo:2013hca,Cachazo:2013iea}. The CHY formula has become a powerful tool for producing amplitudes of a variety of theories, including gravity, in arbitrary dimensions \cite{Cachazo:2014xea,Geyer:2015jch,delaCruz:2015raa,Geyer:2016wjx,Feng:2016nrf}. Moreover, it leads to ways of combining amplitudes of two theories to produce new ones \cite{Cachazo:2013iea} generalizing the  Kaway-Lewellen-Tye (KLT) construction  \cite{Kawai:1985xq} discovered in the 80's. The key ingredient in the CHY reformulation of KLT-like relations is the set of amplitudes of a cubic scalar theory with $U(N)\times U(\tilde N)$ flavor group. The Lagrangian of the theory is given by
\be\label{biadj}
{\cal L} = \partial_\mu \Phi_{a\tilde a}\partial^\mu\Phi^{a\tilde a} + g f^{abc}{\tilde f}^{\tilde a\tilde b\tilde c}\Phi_{a\tilde a} \Phi_{b\tilde b}\Phi_{c\tilde c}
\ee
where $f^{abc}$ and ${\tilde f}^{\tilde a\tilde b\tilde c}$ are the structure constants of the flavor group \cite{Cachazo:2013iea}.

It is well-known that scattering amplitudes of $n$ particles in the adjoint representation of a unitary group can be decomposed into partial amplitudes labeled by a cycle, i.e., a connected 2-regular graph on $n$ vertices \cite{tHooft:1973alw,Berends:1987cv,Mangano:1987xk}. The theory defined by \eqref{biadj} has two unitary groups and therefore its partial amplitudes are labeled by two cycles $C_\alpha$ and $C_\beta$ on $n$ vertices and usually denoted by $m_n(\alpha|\beta)$. Here we choose to make the dependence on the cycles explicit when necessary by writing $m_n(C_\alpha|C_\beta)$.

We attempt to use graph language in a way that is broadly both consistent with graph theory and previous work in the CHY formalism, as will be summarized at the end of the introduction.
The CHY formulation starts by defining a map from the set of 2-regular loopless graphs, including multigraphs, to an $(n-3)!$-dimensional real vector space
\be
\phi: {\cal G}^{\hbox{{\small 2-reg}}} \rightarrow \mathbb{R}^{(n-3)!}.
\ee
We will refer to graphs in the set ${\cal G}^{\hbox{{\small 2-reg}}}$ simply as 2-regular graphs. The map has the following crucial property: Given any pair of 2-regular graphs, $G_1$ and $G_2$, on the same vertex set, the inner product $\phi(G_1)\cdot \phi(G_2)$ only depends on the 4-regular graph obtained by the edge-disjoint union $G_1\cup G_2$. More explicitly, if $G_1\cup G_2$ admits a different decomposition in terms of a pair of 2-regular graphs, i.e., $G_1\cup G_2=G_3\cup G_4$ then $\phi(G_1)\cdot \phi(G_2)=\phi(G_3)\cdot \phi(G_4)$. The amplitudes of the biadjoint theory are then given by $m_n(C_\alpha|C_\beta) = \phi(C_\alpha)\cdot \phi(C_\beta)$.

In \cite{Cachazo:2015nwa}, Gomez and one of the authors noticed that the natural extension $m_n(G_1|G_2)=\phi(G_1)\cdot \phi(G_2)$ can be expressed completely in terms of $m_n(\alpha|\beta)$ if a certain condition is satisfied.

In order to state the condition a definition is needed.  

\begin{defn}\label{comp}
Given a 2-regular graph $G$, a \emph{compatible cycle} to $G$ is a cycle $C$ such that the 4-regular graph obtained by the edge-disjoint union $G\cup C$ admits a hamiltonian decomposition, i.e., $G\cup C = C_1\cup C_2$ where $C_1$ and $C_2$ are both cycles on the same vertex set as $G$.
\end{defn}
See the end of the section for the general definition of a hamiltonian decomposition.

The construction of $m_n(G_1|G_2)$ in terms of $m_n(C_\alpha|C_\beta)$ requires solving the following:

\begin{problem}\label{main problem}
Given a 2-regular graph $G$ on $n$ vertices, find at least $(n-3)!$ compatible cycles such that under $\phi$ they form a basis of $\mathbb{R}^{(n-3)!}$.
\end{problem}

The reason is that if such a basis is found then the vector $\phi(G)$ can be expanded in terms of any basis of cycles, already known to exist, but with coefficients which can be computed entirely in terms of $m_n(C_\alpha|C_\beta)$, by using $\phi(G)\cdot \phi(C)$ with $C$ compatible to $G$ to produce linear equations for the coefficients. In section 2 we provide details on this construction.

In this work we study the combinatorial part of the problem and prove the following theorem.

\begin{thm}
Given a 2-regular graph $G$ on $n$ vertices, there are at least $(n-2)!/4$ compatible cycles for $G$.  In the case that $G$ has only even cycles then there are at least $(n-2)!/2$ compatible cycles for $G$.
\end{thm}

The proof is constructive and provides an algorithm for finding the compatible cycles. Note that $(n-2)!/4 \geq (n-3)!$ for $n\geq 6$ and so while we do not solve problem~\ref{main problem} as we do not have a combinatorial handle on the linear independence, the number of compatible cycles is favorable. For $n<6$ there are also many compatible cycles as computed exactly by one of us with Gomez in \cite{Cachazo:2015nwa}; in particular the explicit computation gives a basis of $\mathbb{R}^{(n-3)!}$ for all $n\leq 6$ cases.

Another reason to be optimistic about the future resolution of the linear independence problem is the work of Bjerrum-Bohr, Bourjaily, Damgaard, and Feng, \cite{Bjerrum-Bohr:2016axv}, in which monodromy relations expressed in terms of cross ratios were used to find an algorithm for the expansion of $\phi(G)$ in term of a basis of cycles, although the coefficients are not manifestly given in terms of $m_n(C_\alpha|C_\beta)$. We give more details on their construction in section 2.

The paper starts in section 2 with a brief review of the Feynman diagram definition of $m_n(\alpha|\beta)$ and the formula for defining $m_n(G_1|G_2)$ which uses the compatible cycles. This section can be skipped in a first reading of the paper in case the reader is only interested in the proof of the result for 2-regular graphs. In section 3, we provide a simple construction which not only gives a lower bound for the number of compatible cycles which is larger than $(n-3)!$  but also an algorithm to find them. In section 4 we establish a connection to breakpoint graphs. We end in section 5 with a short discussion on the issue of finding a basis of $\mathbb{R}^{(n-3)!}$ by using super Catalan numbers to give a lower bound on the number of randomly selected cycles needed to generate a basis of $\mathbb{R}^{(n-3)!}$. This counting indicates that the larger the $n$ the harder it is to find a linear independence basis.  We discuss some modifications to the original algorithm of \cite{Cachazo:2015nwa} and give an outlook with future directions.

\subsection{Review of Graph Theory Terminology}

We end the introduction with a short review of graph theory terminology. Readers are encourage to skip this in a first reading and only use it if needed. 

A graph is loopless if it has no edge with both ends at the same vertex.

For us graphs may have multiple edges (hence being multigraphs in the usual graph theoretic sense), but must be loopless.  

\begin{defn}
A graph is $k$-regular if all vertices have degree $k$, that is, have $k$ edges ending on them.
\end{defn}
We are particularly interested in 2-regular graphs, which are simply a collection of cycles.

As used above given two graphs $G_1$ and $G_2$ on the same vertex set we will write $G_1\cup G_2$ for the graph whose edges are the disjoint union of the edges of $G_1$ and the edges of $G_2$.  In particular if the same edge appears in $G_1$ and $G_2$ then that edge will be a double edge in $G_1\cup G_2$.

\begin{defn}
A hamiltonian cycle in a graph $G$ is a subgraph of $G$ which is a cycle and which uses each vertex of $G$ exactly once.

Given a $2k$-regular graph $G$, a hamiltonian decomposition of $G$, when it exists, is a decomposition of the edges of $G$ into $k$ disjoint hamiltonian cycles: $G=C_1\cup C_2\cup \cdots \cup C_k$ with each $C_j$ a cycle on the same vertex set as $G$. 
\end{defn}

For the main argument we also need the notion of a perfect matching.  
\begin{defn}
A matching in a graph $G$ is a 1-regular subgraph, that is, a subset of edges of the graph where no two edges of the subset share a vertex. 

A perfect matching in a graph $G$ is a matching that uses all vertices of the graph.  We will also use the notion of perfect matching on a vertex set (without the requirement of being a subgraph of some $G$), meaning simply a 1-regular graph on that vertex set.

Given a perfect matching $M$ in a graph $G$, and a vertex $v$ of $G$, the $M$-neighbour of $v$ is the vertex connected to $v$ by an edge of $M$.
\end{defn}

For more graph theory background the reader is referred to \cite{diestel} or \cite{Bondy:2008:GT:1481153}.

\section{Biadjoint scalar amplitudes and extension to general 2-regular graphs}

In this work we are interested in tree-level scattering amplitudes of a quantum field theory of massless scalars interacting via cubic couplings controlled by the structure constants of the algebra of $U(N)\times U(\tilde N)$. The lagrangian presented in \eqref{biadj} produces Feynman diagrams which can be decomposed according to the algebra structure leading to what is known as a color decomposition of amplitudes into partial amplitudes. Consider the scattering of $n$ particles carrying $U(N)\times U(\tilde N)$ labels $\{a_1,\tilde a_1\},\{a_2,\tilde a_2\},\ldots, \{a_n,\tilde a_n\}$, then the amplitude can be written as
\be
A_n(\{a_i,\tilde a_i\}) =\sum_{\alpha,\beta \in S_n/{\mathbb{Z}_n}}{\rm Tr}\left(T^{a_{\alpha(1)}}T^{a_{\alpha(2)}}\cdots T^{a_{\alpha(n)}}\right){\rm Tr}(\tilde{T}^{\tilde a_{\beta(1)}}\tilde{T}^{\tilde a_{\beta(2)}}\cdots \tilde{T}^{\tilde a_{\beta(n)}})\, m_n(\alpha|\beta).
\ee
Here $T^a$ and $\tilde{T}^{\tilde a}$ are the generators of the Lie algebra of $U(N)$ and  $U(\tilde N)$ respectively, i.e., they form a basis of the space of $N\times N$ (or $\tilde N\times \tilde N$) hermitian matrices.

Each particle carries a momentum vector $k_a^\mu$ and $m_n(\alpha|\beta)$ is only a function of Mandelstam invariants $s_{ab}:=2k_a\cdot k_b$. These invariants form a real $n\times n$ symmetric matrix satisfying the following properties
\be\label{con}
s_{aa} = 0\quad {\rm and}\quad \sum_{b=1}^n s_{ab}=0\quad \forall \; a\in \{1,2,\ldots ,n\}.
\ee
The space of kinematic invariants is $n(n-3)/2$ dimensional.

A tree-level Feynman diagram in a cubic scalar theory is defined as a tree with $n$ leaves and $n-2$ trivalent vertices.  We will assume our Feynman diagrams are tree-level from here on out. To each Feynman diagram $\Gamma$ one associates a rational function of $s_{ab}$ as follows. Let $E_\Gamma$ be the set of edges connecting two trivalent vertices. Removing $e\in E_\Gamma$ divides $\Gamma$ into two disconnected graphs with a corresponding partition of the leaves into two sets $L_e\cup R_e=\{1,2,\ldots ,n\}$. The conditions \eqref{con} imply that
\be
\sum_{a,b\in L_e}s_{ab} = \sum_{c,d\in R_e}s_{cd}
\ee
and therefore it is a quantity that can be associated with the edge $e$.

The rational function associated with $\Gamma$ is then
\be
R_\Gamma(S) :=\prod_{e\in E_\Gamma}\left(\sum_{a,b\in L_e}s_{ab}\right)^{-1}.
\ee
There are trivial factors of $2$ generated from the symmetric way the sums in the denominator were defined and can be eliminated if desired.

Any Feynman diagram $\Gamma$ admits several planar embeddings. A planar embedding is a drawing of $\Gamma$ on a disk such that no lines cross and all leaves are attached to the boundary of the disk. Since we are working with trees, any given planar embedding is uniquely specified by the (cyclic) ordering of the labels $\{1,2,\ldots ,n\}$ on the boundary of the disk.

There are $(n-1)!$ possible cyclic orderings, i.e. distributions of $n$ labels on the boundary of a disk. However, it is convenient to identify two orderings if they are related by a reflection. This means that there are only $(n-1)!/2$ inequivalent ones. Let ${\cal O}$ denote the set of all $(n-1)!/2$ orderings. More precisely,
\be\label{defCycles}
{\cal O}:=\{\omega \in S_n: \omega(1)=1, \, \omega(2)<\omega(n) \}.
\ee
The first condition reduces the $n!$ permutations to $(n-1)!$ by using cyclicity to fix $1$ while the second condition selects one of the two permutations related by a reflection that fixes $1$.

\begin{defn}\label{setCy}
Let $\Omega(\omega)$ be the set of all Feynman diagrams with n leaves that a admit a planar embedding defined by $\omega\in {\cal O}$.
\end{defn}

Now we are ready to give a formula for partial amplitudes in terms of Feynman diagrams
\be\label{mdef}
m_n(\alpha| \beta) := (-1)^{w(\alpha,\beta)}\sum_{\Gamma\in \Omega(\alpha)\bigcap\Omega(\beta)}R_\Gamma(S).
\ee

In this formula the sum is over all Feynman diagrams that admit both a planar embedding defined by $\alpha$ and one defined by $\beta$. The overall sign is not is important for the purposes of this work so we refer the reader to \cite{Cachazo:2013iea} for its definition.

In a nutshell, the CHY formulation of $m_n(\alpha| \beta)$ requires finding the critical points of 
\be
{\cal S}(x_1,x_2,\ldots ,x_n) :=\sum_{1\leq a<b\leq n}s_{ab}\, \log (x_a-x_b).
\ee
There are $(n-3)!$ critical points obtained as solutions to what are known as the scattering equations \cite{Cachazo:2013gna,Cachazo:2013hca,Cachazo:2013iea}
\be\label{sceq}
\frac{\partial {\cal S}}{\partial x_a} = \sum_{b=1,b\neq a}\frac{s_{ab}}{x_a-x_b}=0\quad \forall\; a\in \{1,2,\ldots,n\}.
\ee
Let's denote the $(n-3)!$ solutions as $x_a^{I}$. In general the solutions are complex but when the $s_{ab}$'s are chosen in what is known as the positive region all solutions are real \cite{Cachazo:2016ror}. Given any cycle $C_\alpha$, one constructs a vector $\phi(C_\alpha)\in\mathbb{R}^{(n-3)!}$ whose components are given by
\be\label{defPhi}
\phi(C_\alpha)_I := \frac{K_I}{(x_{\alpha_1}^I-x_{\alpha_2}^I)(x_{\alpha_2}^I-x_{\alpha_3}^I)\cdots (x_{\alpha_n}^I-x_{\alpha_1}^I)},
\ee
where $K_I$ is a function obtained from second derivatives of ${\cal S}$ and is invariant under permutations of labels and hence $\alpha$ independent. Therefore $K_I$ is not relevant to our discussion and we refer the reader to \cite{Cachazo:2013iea} for details.  

Finally, partial amplitudes are computed as
\be\label{weq}
m_n(\alpha| \beta) = \sum_{I=1}^{(n-3)!}\phi(C_\alpha)_I\,\phi(C_\beta)_I.
\ee
We will also use the notation $\phi(C_\alpha)\cdot\phi(C_\beta)$ for the inner product in \eqref{weq}.

Now it is clear how to generalize $\phi$ to a map that assigns to any 2-regular graph a vector in $\mathbb{R}^{(n-3)!}$. Let $G$ be any 2-regular graph with edge set $E$ then
\be
\phi(G)_I:= K_I\prod_{e\in E}\frac{1}{x_{e_i}^I-x_{e_f}^I}.
\ee
Given any two 2-regular graphs $G_1$ and $G_2$ one also defines
\be
m_n(G_1|G_2):=\phi(G_1)\cdot \phi(G_2).
\ee
As mentioned in the introduction the map $\phi$ has the property, which is clear from its definition, that the value of $m_n(G_1|G_2)$ is only a function of the 4-regular graph obtained as the union $G_1\cup G_2$. 

The scattering equations \eqref{sceq} are polynomial equations and are difficult to solve for generic values of $s_{ab}$. This is why it is useful to try and express $m_n(G_1|G_2)$ in terms of $m_n(\alpha|\beta)$, which are known rational functions of $s_{ab}$. One way to achieve this was proposed by Gomez and one of the authors in \cite{Cachazo:2015nwa}. The first step is to choose any basis of $\mathbb{R}^{(n-3)!}$ made out of vectors corresponding to cycles, not necessarily compatible to any $G_i$. For example, it is known that by fixing the position of three labels and permuting the rest one has $(n-3)!$ cycles that generate a basis (see e.g. \cite{Cachazo:2015nwa}). Consider one such sets ${\cal A}= \{(\gamma ,n-2,n-1,n): \gamma \in S_{n-3}\}$ and expand $\phi(G_i)$ in the corresponding basis
\be\label{coeff}
\phi(G_i)=\sum_{\alpha\in {\cal A}} c_{i,\alpha}\phi(C_\alpha).
\ee 
Now, if a basis ${\cal B}_i$ of $\mathbb{R}^{(n-3)!}$ is found using compatible cycles to $G_i$ then it is possible to compute the coefficients $c_{i,\alpha }$ by solving the system of equations
\be\label{expand}
\phi(G_i)\cdot \phi(C_\beta)=\sum_{\alpha\in {\cal A}} c_{i,\alpha}\,\phi(C_\alpha)\cdot \phi(C_\beta)
\ee 
with $C_\beta$ in ${\cal B}_i$. Therefore $\phi(G_i)\cdot \phi(C_\beta) = \phi(C)\cdot \phi(C')$ for some cycles $C$ and $C'$.   

Using \eqref{coeff} one finds that
\be
m_n(G_1|G_2)= \sum_{\alpha,\beta\in {\cal A}}c_{1,\alpha}c_{2,\beta}\,m_n(\alpha,\beta)
\ee
and since all coefficients $c_{i,\alpha}$ are known using \eqref{expand} we have achieved the desired formula.

Let us end this section with a short description of the algorithm from \cite{Bjerrum-Bohr:2016axv} mentioned in the introduction which also achieves an expansion of the form \eqref{coeff} with coefficients given in terms of the invariants $s_{ab}$. The main tool is the monodromy relations expressed in terms of cross ratios \cite{Cardona:2016gon}: For any subset $A\subset \{1,2,\ldots ,n\}$ with $2\leq |A|\leq n-2$ and for any $a\in A$ and $b\in A^{\rm c}=\{1,2,\ldots ,n\}\setminus A$, 
\be\label{mono}
1= -\sum_{c\in A,d\in A^{\rm c}}s_{cd}\frac{(x_a-x_c)(x_d-x_b)}{(x_b-x_c)(x_a-x_d)}.
\ee
In \cite{Bjerrum-Bohr:2016axv} the identity \eqref{mono}, which holds on the support of the scattering equations, is used to write $\phi(G)$ as a linear combination of $2$-regular graphs with less cycles. In other words, \eqref{mono} can be used to fuse cycles. Iterating the procedure until all graphs involved are single cycles gives rise to an expansion of the form \eqref{coeff}, although the coefficients are not manifestly given in terms of $m_n(C_\alpha|C_\beta)$. It would be interesting to try and find a connection between the two kinds of expansions.

\section{Lower bounds on the number of compatible cycles}

In the following arguments we will use the notion of perfect matching in a slightly different way than what is typical in graph theory.  Given a graph $G$, when we refer to a perfect matching on $V(G)$, the vertex set of $G$, we mean any 1-regular graph on the vertices $V(G)$.  So this is not a perfect matching \emph{of} $G$ as it is not a subgraph of $G$.  It is simply a perfect matching of the complete graph on $V(G)$, in other words a 1-regular graph on $V(G)$.

Additionally, to reiterate what was mentioned in the introduction, for us graphs are loopless but can have multiple edges.  Furthermore, when we take the union of two graphs on the same vertex set, this denotes the disjoint union on the edge sets.  That is, if $G_1$ and $G_2$ are graphs on the same vertex set $V$, and both $G_1$ and $G_2$ have one edge between $v_1$ and $v_2$, $v_1, v_2\in V$, then $G_1\cup G_2$ has two edges between $v_1$ and $v_2$.

With this in mind we are ready to count compatible cycles.  We begin by counting compatible cycles to graphs with only even length cycles.

\begin{thm}\label{alleven}
Let $G$ be a 2-regular graph on $n$ vertices which consists of only even cycles.  There are at least $(n-2)!/2$ compatible cycles for $G$.  
\end{thm}

 For $G$ as in the statement of the theorem, we will fix a decomposition of $G$ as $G = A \cup B$ with $A$ and $B$ perfect matchings on $G$, so that $G$ consists of $AB$-alternating cycles.
With this decomposition in mind we will prove the theorem with the help of two lemmas, as follows.  The first lemma simply counts how many ways there are to complete a perfect matching into a cycle.

\begin{lem}\label{lem P}
Let $A$ be a perfect matching on $n$ vertices.  There are $(n-2)!!$ perfect matchings $P$ on the same set of vertices such that $P \cup A$ is a cycle.
\end{lem}
\begin{proof}
Pick a vertex $v \in V(A)$. Let its neighbour in $A$ be $a$.  Starting from $v$, there are $n-2$ choices for a vertex $p$ which is distinct from $a$ and $v$.  Let $vp$ be an edge in $P$. Now let $a'$ be the neighbour of $p$ in $A$, and there are $n-4$ choices for a vertex $p'$ which is distinct from $v, a, p, a'$ which we also add to $P$.  Continuing likewise, we can extend the path $v,a,p,a',p',\ldots$.  For a final edge of $P$ take the edge from the end of the path (which by construction will not yet have an incident $P$-edge) to $v$.  Then $P$ is a perfect matching, $P\cup A$ is a cycle and there are $(n-2)!!$ choices for $P$ constructed in this manner.  Furthermore, all $P$ as in the statement can be constructed in this manner, as the cycle $P\cup A$ determines the choices.
\end{proof}

Another way to prove the previous lemma is to contract $A$, pick a cycle on the remaining $n/2$ vertices, and then note that this cycle can be expanded back to the original vertex set to give a $P$ as in the statement in $2^{n/2-1}$ ways, because after inserting the first edge of $A$ into the cycle, each remaining edge of $A$ can be inserted into the cycle in one of two ways.  Then since $(n-2)!! = 2^{n/2-1}(n/2-1)!$ for even $n$ we obtain the same result.

\begin{lem}\label{lem Q}
Let $B$ and $P$ be two perfect matchings on the same set of $n$ vertices.  Then there are at least $(n-3)!!$ choices for a perfect matching $Q$ on this vertex set with the property that both $Q \cup B$ and $P \cup Q$ are cycles.
\end{lem}
\begin{proof}
This proof is the main part of the whole argument for the general result. We proceed by induction. The base case is $n=4$ and the result follows by checking several cases: either $B=P$ on 4 vertices or $B\cup P$ is a cycle on 4 vertices. Since $n=4$ we only need to find $1=(4-3)!!$ perfect matching $Q$ with the desired properties.

In either case, we can simply draw at least one $Q$ no matter the choice of $P$, as illustrated in figure~\ref{fig all even base}.
\begin{figure}[h]
\centering
\includegraphics[scale=1]{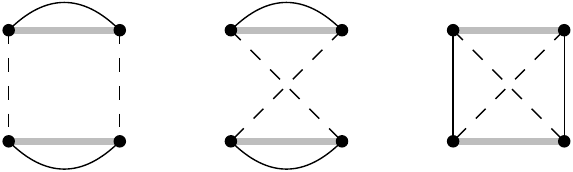}
\caption{The base case for the lemma. Thick grey edges are $B$, thin black edges are $P$, and thin dashed edges are the possibilities for $Q$.}\label{fig all even base}
\end{figure}

For the induction, let $V$ be the vertex set.  Pick a vertex $v$. Label its $B$-neighbour $b$ and its $P$ neighbour $p$. Pick a vertex $q$ not in $\{v,b,p\}$ and draw a $Q$-edge $vq$. There are at least $n-3$ choices for this $q$ (there may be more than $n-3$ choices as $v$, $b$, $p$ may not all be distinct).  Given such a choice of $q$, label its $B$-neighbour $\beta$ and its $P$-neighbour $\pi$. From here, we create two matchings on the vertex set $V' = V\setminus \{v,q\}$, namely $B' = (B|_{V'}) \cup \{b\beta\}$ and $P' = (P|_{V'}) \cup \{p\pi\}$. It should be noted that these are indeed matchings, since none of $b, p, \beta, \pi$ are saturated in the restrictions of their respective matchings to $V'$, and all of these matchings are also perfect. Now, $B'$ and $P'$ satisfy the induction hypothesis, and so give rise to $(n-2-3)!!$ choices of $Q'$ with the property that both $C'_B:= Q' \cup B'$ and $C'_P := P' \cup Q'$ are cycles.

The goal from here is to lift $Q'$ up to the perfect matching $Q := Q' \cup \{vq\}$ on $V$ and show that $Q$ satisfies the lemma. To this end, note that $C'_B\setminus \{b\beta\}$ and $C'_P\setminus \{p\pi\}$ induce paths $S_B$ and $S_P$ on $V$ which hit all of the vertices except $v$ and $q$. Therefore $S_B \cup \{bv, vq, q\beta\}$ is a cycle consisting of all of the edges of $B$ and $Q$, and $S_P \cup \{pv, vq, q\pi\}$ is a cycle consisting of all the edges of $P$ and $Q$. This means $Q$ satisfies the lemma.

Now, there were at least $n-3$ choices for the edge $vq$ and at least $(n-5)!!$ choices for the matching $Q'$. If there is no repetition here, we will have at least $(n-3)!!$ choices for $Q$ and the claim will be proven. To see that there is indeed no repetition, note that two different choices of $q$ cannot lead to the same cycle, and given the same choice of $q$, the paths $S_B$ and $S_P$ will depend only on the (already distinct) choices of $Q'$. This completes the proof.
\end{proof}

\begin{proof}[Proof of Theorem \ref{alleven}.]
Decompose $G$ as $G = A \cup B$ with $A$ and $B$ perfect matchings on $G$, so that $G$ consists of $AB$-alternating cycles.

By the first lemma we have $(n-2)!!$ perfect matchings $P$ such that $P\cup A$ is a cycle.  For each such $P$ then apply the second lemma to obtain $(n-3)!!$ perfect matchings $Q$ such that $Q\cup B$ and $P\cup Q$ are also cycles.

The compatible cycle thus constructed is $P\cup Q$, but each such compatible cycle can potentially appear twice as either of the two perfect matchings making it up could have been constructed first. The result, then, is at least $$\frac{(n-2)!!(n-3)!!}{2} = \frac{(n-2)!}{2}$$ compatible cycles as desired.
\end{proof}

\begin{figure}[h]
\centering
\includegraphics{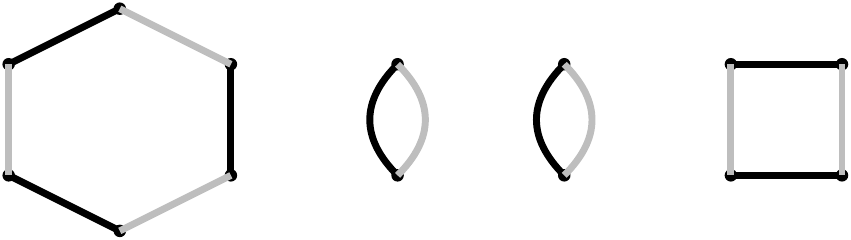}
\caption{An example graph $G$ with all even cycles decomposed as $A\cup B$ where thick black edges are $A$ and thick grey edges are $B$.}\label{fig even alg eg G}
\end{figure}

To illustrate how this theorem can be used algorithmically to construct compatible cycles consider the example graph in Figure~\ref{fig even alg eg G}.  By the first lemma we can construct the perfect matching $P$ by beginning at a vertex, say the upper of the two leftmost vertices in the figure, following $A$, in this case to the top vertex, and then choosing any vertex other than the two already mentioned to join to the top vertex making an edge for $P$.  Suppose we choose the lower of the two vertices to the right in the same cycle of $G$.  Then we follow $A$ again and pick any vertex not already seen to add a new edge to $P$ and so on.  Continuing in this way one possible $P$ we could obtain is as illustrated in Figure~\ref{fig even alg eg G and P}.

\begin{figure}[h]
\centering
\includegraphics{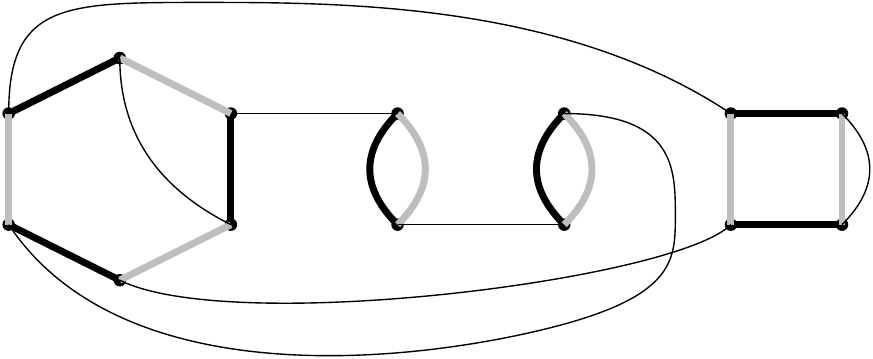}
\caption{The example graph $G$ along with a perfect matching $P$ (thin black lines) so that $P\cup A$ is a cycle.}\label{fig even alg eg G and P}
\end{figure}

\begin{figure}[h]
\centering
\includegraphics{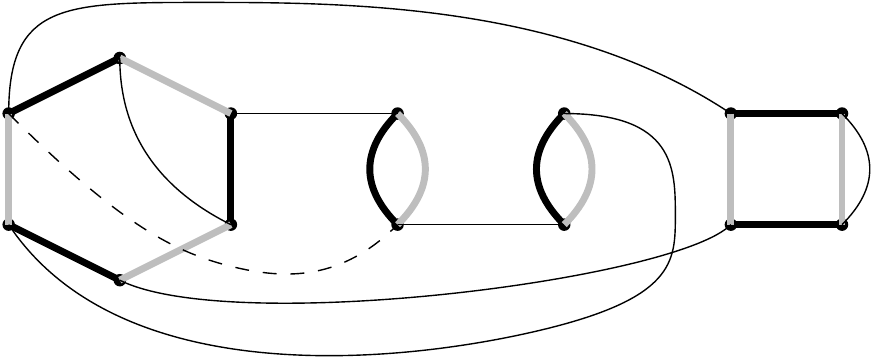}
\caption{$G$ and $P$ along with a first edge for the construction of $Q$ (dashed line).}\label{fig even alg eg G and first Q}
\end{figure}

Next we follow the second lemma.  Beginning again at the upper of the two leftmost vertices, we pick any vertex other than this vertex's neighbours in $B$ and $P$ to make an edge for $Q$.  In this case say we pick the lower vertex of the leftmost bubble.  This is illustrated in Figure~\ref{fig even alg eg G and first Q}.

\begin{figure}[h]
\centering
\includegraphics{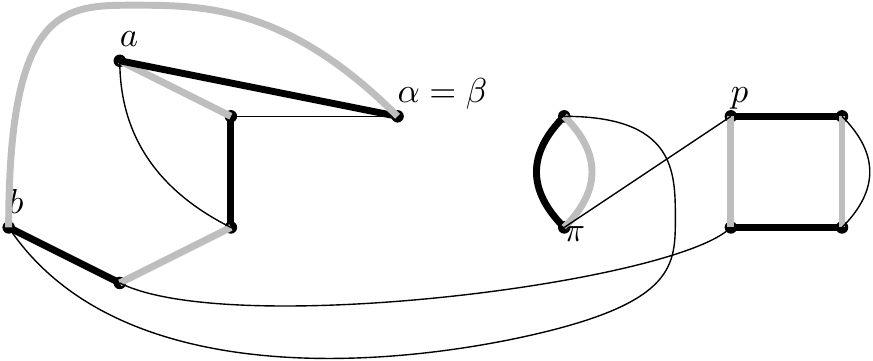}
\caption{The graph $G' = A'\cup B'$ with $P'$.}\label{fig even alg eg Gp}
\end{figure}

{}From this choice of edge the second lemma tells us to construct $G'$ (along with $P'$) as illustrated in Figure~\ref{fig even alg eg Gp}, with vertex labels as in the lemma.  

\begin{figure}[h]
\centering
\includegraphics{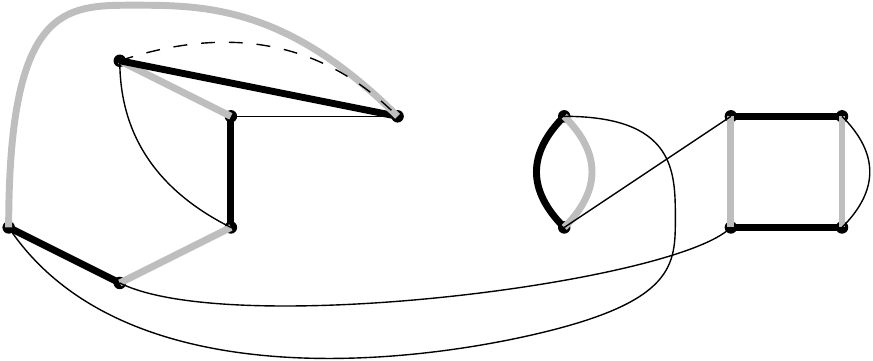}
\caption{The graph $G' = A'\cup B'$ with $P'$ and a first choice of edge for $Q'$.}\label{fig even alg eg Gp and first Q}
\end{figure}

\begin{figure}[h]
\centering
\includegraphics{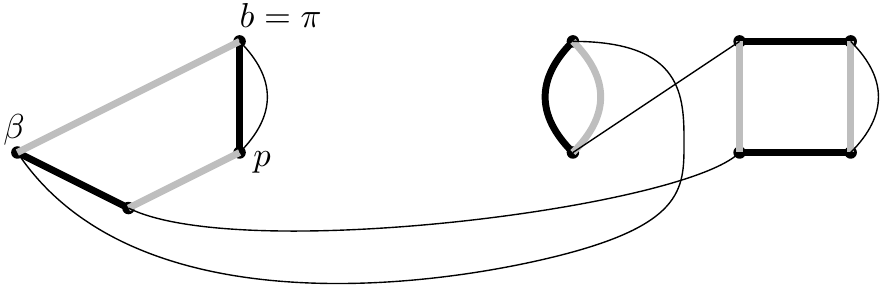}
\caption{The graph $G'' = A''\cup B''$ with $P''$.}\label{fig even alg eg Gpp}
\end{figure}

\begin{figure}[h]
\centering
\includegraphics{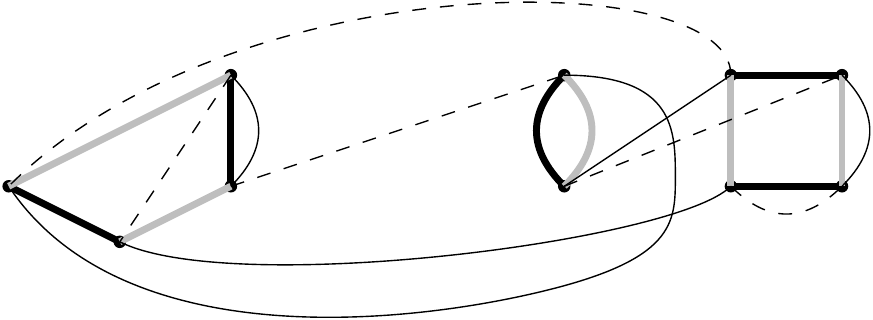}
\caption{The graph $G'' = A''\cup B''$ with $P''$ and $Q''$.}\label{fig even alg eg Gpp and Q}
\end{figure}

The process now continues.  Let's progress one more step explicitly, choosing the first edge of $Q'$ as shown in Figure~\ref{fig even alg eg Gp and first Q}.  This results in the graph $G''$ as illustrated in Figure~\ref{fig even alg eg Gpp}.  Continuing the process we can construct $Q''$; one possibility for $Q''$ is illustrated in Figure~\ref{fig even alg eg Gpp and Q}

\begin{figure}[h]
\centering
\includegraphics{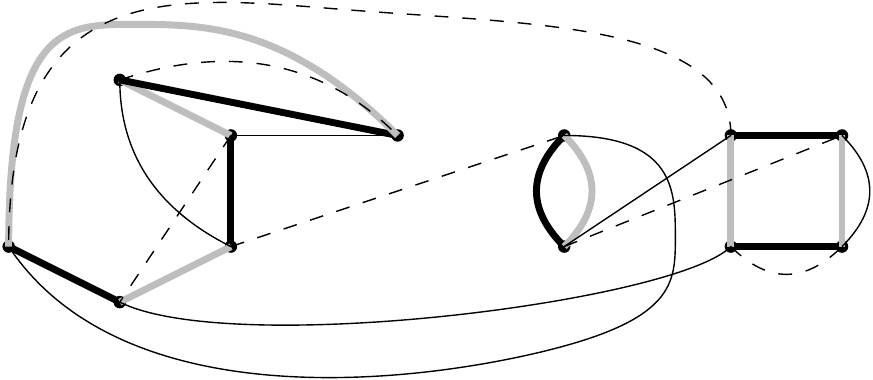}
\caption{The graph $G' = A'\cup B'$ with $P'$ and $Q'$.}\label{fig even alg eg Gp and Q}
\end{figure}

\begin{figure}[h]
\centering
\includegraphics{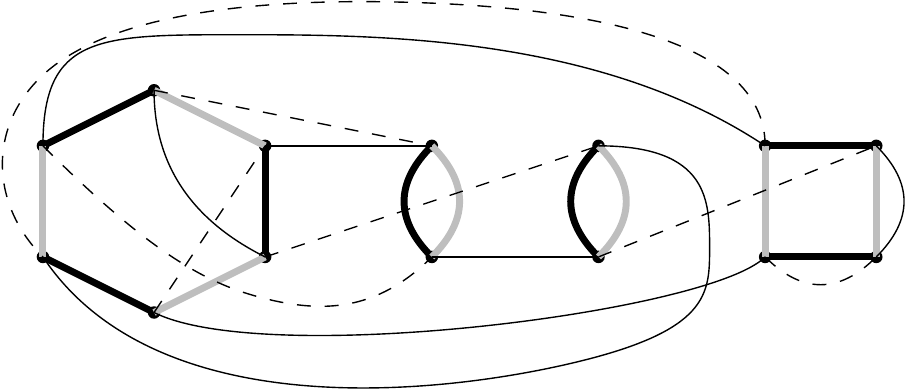}
\caption{The graph $G = A\cup B$ with $P$ and $Q$.  The thick black edges are $A$, the thick grey edges are $B$, the thin black edges are $P$ and the dashed edges are $Q$.}\label{fig even alg eg G and Q}
\end{figure}

Bringing $Q''$ up to $Q'$ on $G'$ we obtain the situation illustrated in Figure~\ref{fig even alg eg Gp and Q}, and bringing $Q'$ up to $Q$ on $G$ we obtain our compatible cycle $P\cup Q$ as illustrated in Figure~\ref{fig even alg eg G and Q}.  Observe that in this last figure, $A\cup P$, $B\cup Q$ and $P\cup Q$ are all cycles as expected.

As this example illustrates, the theorem in fact gives an algorithm to generate at least $(n-2)!/2$ compatible cycles for any 2-regular $G$ with all cycles even.

\begin{thm}
For an arbitrary $2$-regular graph $G$, there are at least $(n-2)!/4$ compatible cycles $C$ for $G$.  In the special case where $G$ has only even cycles then there are at least $(n-2)!/2$ compatible cycles.
\end{thm}
\begin{proof}
If $G$ has only even cycles, then apply the previous result.  Now assume $G$ has at least one odd cycle.

Let $O_1, \hdots, O_k$ be the odd cycles of $G$. Pick a vertex $v_i$ from each $O_i$. Now we will `bandage' these cycles at the $v_i$ in the sense that the $v_i$ will be treated just as a point along the `single edge' between their neighbours. Formally, define $G'$, the bandaged graph, to be the graph obtained from $G$ by contracting one of the edges incident to $v_i$ for each $1\leq i\leq k$; we no longer use the vertex labels $v_i$ in $G'$, as we think of the contracted vertices as having come from their other vertex in $G$ while $v_i$ is gone as it has been bandaged up.  Additionally let $e_i$ be the edge in $G'$ which came from the non-contracted incident edge to $v_i$ for each $1\leq i \leq k$.  We can then obtain $G$ from $G'$ by in each edge $e_i$ putting a new vertex $v_i$.

Applying the previous theorem we have $(n-k-2)!/2$ choices for $C$ on $G'$. Now we must extend $C$ to $G$. We would like to do this by, for each $v_i$ in turn, picking an edge $ww'$ of $C$ and replacing it with the edges $wv_i$ and $v_iw'$.  As we do so we increase the number of edges in $C$ and so increase the number of ways to continue this process in subsequent steps.

However, not every choice will preserve that $C$ is a compatible cycle.  Consider then a compatible cycle $C$ for $G'$ given by the previous theorem.  From that construction we have that $G'=A\cup B$ where $A$ and $B$ are perfect matchings (alternating along the cycles of $G'$), and $C=P\cup Q$ where $P$ and $Q$ are perfect matchings such that $P\cup A$ and $Q\cup B$ are also cycles.  Consider now $e_1$.  If $e_1\in A$ then for any $ww'$ in $Q$ we can add $v_1$ to $e_1$, letting both of the resulting edges be in $A$, and we can replace $ww'$ in $Q$ by $wv_1$ and $v_1w'$.  Then with these changes we still have $P\cup Q$, $P\cup A$ and $Q\cup B$ cycles.  Note that if $e_1\in A$ but $ww'$ were in $P$ then the same construction would result in $P\cup A$ not being a cycle.  However, $e_1\in B$ and $ww'\in P$ also results in $P\cup Q$, $P\cup A$ and $P\cup B$ remaining cycles.  Consequently we have $(n-k)/2$ choices for $ww'$.

Continuing with $e_2, e_3, \ldots$, the argument above did not require that $A,B,P,Q$ were matchings, and so whenever $e_i \in A$ we take $ww'\in Q$ and whenever $e_i\in B$ we take $ww'\in P$.  Also, whenever $ww'\in Q$ then $Q$ has one more edge after that step of the construction and whenever $ww'\in P$ then $P$ has one more edge after that step of the construction.  All together, then, we have
\[
\underbrace{\left(\frac{n-k}{2}\right)\left(\frac{n-k}{2}+1\right)\left(\frac{n-k}{2}+2\right)\cdots}_{\text{as many factors as $e_i\in A$}}\underbrace{\left(\frac{n-k}{2}\right)\left(\frac{n-k}{2}+1\right)\left(\frac{n-k}{2}+2\right)\cdots}_{\text{as many factors as $e_i\in B$}}
\]
choices to extend $C$ to a compatible cycle on $G$.  The expression above is bounded below by
\[
\frac{1}{2^k}\underbrace{(n-k)(n-k)(n-k+2)(n-k+2)\cdots}_{k \text{ times}} \geq \frac{1}{2^k} (n-k-1)(n-k)(n-k+1)(n-k+2)\cdots (n-2)
\]
Combining this with the number of choices for $C$ we get a total of at least
\[
\frac{1}{2^k} (n-k-1)(n-k)(n-k+1)(n-k+2)\cdots (n-2) \frac{(n-k-2)!}{2} =\frac{1}{2^{k+1}}(n-2)!
\]
compatible cycles for $G$.

To take care of the powers of 2, we need to more closely analyze the freedom we had in the initial choices.  Keeping the $v_i$ fixed, note that in the initial choice of decomposition of $G'$ into $A$ and $B$, each $v_i$ is either in $A$ or in $B$.  Let us fix a choice of $A$ and $B$ for $G'$ with $e_1 \in A$.  Suppose we have a compatible cycle $C$ for $G$ constructed as above based on this choice of $A$ and $B$.  Then the edges of $C$ alternate between $P$ and $Q$ except at the $v_i$ where two edges in the same set occur consecutively.  Since we know $e_1 \in A$ in $G'$, the construction above gives that $v_1$ is between two $Q$ edges in $C$.  Following the alternation of edges around $C$, starting with the $Q$-edges around $v_1$ we can determine for each $v_i$ whether it is surrounded by $P$-edges or $Q$-edges.  If a $v_i$ is surrounded by $P$-edges in $C$ then $e_i$ is a $B$ edge in $G'$ and if $v_i$ is surrounded by $Q$ edges in $C$ then $e_i$ is an $A$ edge in $G'$.

The argument of the previous paragraph implies that knowing a compatible cycle $C$ constructed as described above and knowing that $e_1\in A$ in $G'$ is enough to determine which $e_i$ are in $A$ and which in $B$ as edges in $G'$.  However, which $e_i$ are in $A$ and which are in $B$ comes from our initial choice of decomposition of $G'$ into $A$ and $B$.  Consequently, different choices of how the $e_i$ are assigned to $A$ and $B$ must give different compatible cycles $C$.  Since the argument required us to fix $e_1\in A$, it remains to choose which of $A$ or $B$ for the $e_i$ for $2\leq i \leq k$.  That is, there remain $k-1$ binary choices.  

Together with the construction given above for $C$, this means that we obtain a total of at least 
\[
\frac{2^{k-1}}{2^{k+1}} (n-2)! = \frac{(n-2)!}{4}
\]
compatible cycles for $G$.  

\end{proof}

\section{Connection to breakpoint graphs}

Counting compatible cycles is closely related to counting breakpoint graphs, which are certain graphs used in studying genomic rearrangements.  We will not need the definition of a breakpoint graph here (originally due to Bafna and Pevzner \cite{BPbreakpoint}), rather we consider the set up of Grusea and Labarre \cite{GLbreakpoint} which contains a reformualtion of the notion of breakpoint graph that already puts the problem closer to compatible cycle enumeration.

We need the following 
\begin{defn}
\mbox{}
    \begin{itemize}
        \item (\cite{GLbreakpoint} definition 5.2) Given vertices $\{0, 1, \ldots, 2m, 2m+1\}$, a \emph{configuration} is the union of two perfect matchings on those vertices, $\delta_B$ and $\delta_G$ where $\delta_G = \{\{2i, 2i+1\}: 0\leq i \leq m\}$
        \item (\cite{GLbreakpoint} definition 5.3) Given a configuration $\delta_B\cup \delta_G$, write $\overline{\delta}_G$ for the perfect matching $\overline{\delta}_G = \{\{2i-1, 2i\}: 1\leq i \leq m\} \cup \{2m+1, 0\}$ and let the \emph{complement} of the configuration $\delta_B\cup \delta_G$ be $\delta_B\cup \overline{\delta}_G$.
        \item (\cite{GLbreakpoint} definition 3.1) The \emph{signed Hultman number} $\mathcal{S}^\pm_H(m,k)$ is the number of signed permutations on $m$ elements whose breakpoint graph consists of $k$ disjoint cycles.
    \end{itemize}
    \end{defn}  
Then in place of the definition of a breakpoint graph for a signed permutation, we can use the following lemma.
\begin{lem}[\cite{GLbreakpoint} lemma 5.1]
A configuration $\delta_B \cup \delta_G$ is the breakpoint graph of a signed permutation $\pi$ if and only if $\delta_B\cup \overline{\delta}_G$ is a cycle.
\end{lem}
We also don't need the definition of a signed permutation, but merely the observations from \cite{GLbreakpoint} that a signed permutation on $m$ elements leads to a breakpoint graph on $\{0, 1, \ldots, 2m, 2m+1\}$, and the the map between signed permutations and breakpoint graphs is bijective.

\medskip

With this set-up, consider the all-bubbles case of our problem from the previous sections.  That is suppose $G$ consists of $n/2$ double edges which are vertex disjoint.  Label the vertices of $G$ as $\{0, 1, \ldots, 2m, 2m+1\}$ (where $m = (n-2)/2$) so that the double edges of $G$ run between $2i-1$ and $2i$ for $1\leq i \leq m$ and between $2m+1$ and $0$.  Then with notation as above $G= \overline{\delta}_G\cup \overline{\delta}_G$, where as usual the union denotes an edge-disjoint union, so taking the union of two copies of $\overline{\delta}_G$ gives double edges.
 
Now take any breakpoint graph with one cycle and call it $C$.  By Grusea and Labarre's lemma $C$ can be written as $\delta_B \cup \delta_G$ and $\delta_B\cup \overline{\delta}_G$ is a cycle.  In fact $C$ is a compatible cycle for $G$.  To see this note that we have $\delta_B\cup \overline{\delta}_G$ and $\delta_G\cup \overline{\delta}_G$ are both cycles, as is $\delta_B\cup \delta_G$ since we took a breakpoint graph with one cycle.  

Furthermore, all compatible cycles when $G$ consists of only bubbles can be obtained in this way because Grusea and Labarre's lemma says that a breakpoint graph with one cycle is exactly a graph where the above unions of matchings are cycles. 

According to Grusea and Labarre's results the number of such breakpoint graphs is $\mathcal{S}_H^{\pm}(n/2-1, 1)$.  Finally, we need to consider how many different labellings of $G$ would result in different families of breakpoint graphs.  This is asking, given $\overline{\delta}_G$ how many different $\delta_G$ could it correspond to.  This is exactly the problem solved in Lemma~\ref{lem P}, so there are $(n-2)!!$ choices.  As in the compatible cycle construction, this counts each compatible cycle twice since either of the two matchings making it up could be $\delta_G$.  All together this tells us that the number of compatible cycles to a graph $G$ consisting of $n/2$ bubbles is
\be\label{numB}
    \frac{1}{2}(n-2)!!\,\mathcal{S}_H^{\pm}(n/2-1,1).
\ee
In \cite{Cachazo:2015nwa} this formula was guessed based on explicit computations of initial terms along with the sequence A001171 in the OEIS \cite{OEIS}, but now, by the above, it is proven.

Note that this is better than our results of the previous section for the all bubbles case because it gives an exact enumeration.  For more general 2-regular $G$ with only even cycles there remains a connection to breakpoint graph enumeration, but it does not capture all possible compatible cycles.

To explore this more general situation, let $G$ be a 2-regular graph on $n$ vertices with $k$ cycles, all of even length.  Fix a decomposition of $G$ into two matchings.  By Grusea and Labarre's lemma, labelling $G$ so that it is a breakpoint graph is equivalent to finding a third matching which gives a cycle when combined with either matching from $G$.  By Lemma~\ref{lem Q} there are at least $(n-3)!!$ ways to do this.  Fix a labelling of $G$ so that $G$ is a breakpoint graph, and write $G=\delta_B \cup \delta_G$.  Now consider any breakpoint graph $H$ with one cycle (relative to the same labelling).  Then $H=\delta_{B'} \cup \delta_G$ and $\delta_{B'}\cup \overline{\delta}_G$ is a cycle.  Furthermore $\delta_B\cup \overline{\delta}_G$ is a cycle since $G$ is a breakpoint graph and $\delta_{B'}\cup \delta_G$ is a cycle since $H$ was chosen to have only one cycle.  So $C=\delta_{B'}\cup \overline{\delta}_G$ is a compatible cycle for $G$.

Not all compatible cycles for $G$ arise by breakpoint graphs.  However, all the ones constructed by the techniques of the previous section do arise from breakpoint graphs.  Despite this, we do not obtain an improved bound from Grusea and Labarre's results on signed Hultman numbers because when working with a fixed $G$, we are fixing not just $k$, the number of cycles, but also the lengths of each cycle.  This suggests enumerating a refined version of the signed Hultman numbers which keeps track of the cycle structure rather than just the number of cycles.  This could be interesting as combinatorics, and might yield better bounds on compatible cycles, or perhaps applications in breakpoint graphs.

\subsection{Lower bound vs. Exact count}

A natural question is to get an approximate notion of how close our bound is to the actual number of compatible cycles. While finding the exact number seems to be a difficult problem, our lower bound was obtained by using very simple constructions. Luckily, returning to the all bubble case for the moment, a formula for computing the number of breakpoint graphs with one cycle is given in \cite{GLbreakpoint} and therefore we can use it to compare it to our lower bound. Let $s=n/2$ be the number of double edges, i.e., bubbles. The formula for the number of breakpoint graphs $\mathcal{S}_H^{\pm}(s-1,1)$ given by
\be
\mathcal{S}_H^{\pm}(s-1,1) = \frac{2^{(3s-2)}s!(s-1)!^2}{(2s)!} + \sum_{a=1}^{s-1}(-1)^{s+1}s\!\!\!\sum_{b=1}^{{\rm min}(a,s-a)} (-1)^{a-b}T_{a, b, s}
\ee
where
\be
T_{a,b,s}:= \frac{2^{3(a-b)-1}(2a-2b+1)(a-1)!\left((2b)!(a-1)!(s-a-b+1)!\right)^2}{(s^2-(a-b+1)^2)(s^2-(a-b)^2)(s-a-b)!(2a-1)!(b-1)!\left((2b-1)b!\right)^2}.
\ee
In \cite{Cachazo:2015nwa}, the asymptotic behavior of $\mathcal{S}_H^{\pm}(n/2-1,1)$ was numerically studied and found to give the following number of compatible cycles
\be
\frac{1}{2}(n-2)!!\mathcal{S}_H^{\pm}(n/2-1,1)\sim \frac{\pi}{4}n(n-3)!.
\ee 
This means that for the all bubbles case the ratio of the exact count in the asymptotic regime to our lower bound, i.e. $(n-2)!/2$, is only $\pi/2\sim 1.57$.  

\section{Discussion}

In this work we presented a constructive proof of the existence of $(n-2)!/4$ compatible cycles to any 2-regular graph $G$. Moreover, when $G$ possesses only even cycles our lower bound becomes $(n-2)!/2$. Our construction has important applications in the computation of CHY integrals, which give rise to the map $\phi$ as review in section 2. While integrations involving two cycles $C_{\alpha}$ and $C_{\beta}$ compute amplitudes $m_n(\alpha|\beta)$ in a biadjoint scalar theory, more general CHY integrals are known to compute amplitudes in many other theories such as Yang-Mills and Einstein gravity \cite{Cachazo:2013hca}. These more general amplitudes require the integration of functions which are not of the simple form $m_n(\alpha|\beta)$. The more general integrals are associated with arbitrary 2-regular graphs, say $G_1$ and $G_2$, and the corresponding integration, $m_n(G_1|G_2)$, has to be performed.

Several techniques have been proposed in the literature for computing CHY integrals as the ones arising from $m_n(G_1|G_2)$. Some of them use the global residue theorem \cite{Baadsgaard:2015voa}, cross ratio identities \cite{Zhou:2017mfj}, deformations of the scattering equations \cite{Gomez:2016bmv}, etc. The technique relevant to our work expresses $m_n(G_1|G_2)$ directly in terms of the simple objects $m_n(\alpha|\beta)$ and requires finding $(n-3)!$ compatible cycles to $G_i$ such that under the map $\phi$ they generate a basis of $\mathbb{R}^{(n-3)!}$.

We have concentrated on the combinatorial part of the problem leaving the question of linear independence for the future. However there are some comments that can be made which follow from the Feynman diagram interpretation of $m_n(\alpha|\beta)$ and which show that the problem of independence is non-trivial even though our lower bound shows that for large $n$ the number of compatible cycles to any 2-regular graph is, at least, $n/4$ times larger than the size of the required basis. 

\subsection{Linear independence}

In order to show that the problem is non-trivial, let us consider a given cycle; without loss of generality choose the one defined by the canonical order, $C_{\mathbb{I}}$. We want to determine the total number of cycles such that the corresponding vectors under the map $\phi$ are orthogonal to $\phi(C_{\mathbb{I}})$. Let us denote the set of such cycles by $\ort (\mathbb{I})$. More explicitly,
\be
\ort (\mathbb{I}) := \{ C_\alpha \in {\cal O}\, : \, \phi(C_\alpha)\cdot \phi(\mathbb{I})=0 \}.
\ee
Recall that ${\cal O}$, defined in \eqref{defCycles}, denotes the set of all cycles. The reason $\ort (\mathbb{I})$ is interesting is that no subset of cycles in $\ort (\mathbb{I})$, including the whole set, can possibly give a basis of $\mathbb{R}^{(n-3)!}$. This is clear as they would not be able to generate the vector $\phi(\mathbb{I})$ which is orthogonal to that space.

Let us determine the size of $\ort (\mathbb{I})$. Start by recalling that there are $(n-1)!/2$ cycles for $n$ labels and that $m_n(\mathbb{I}|\alpha) =  \phi(\mathbb{I})\cdot \phi(C_\alpha)$ can be computed using Feynman diagrams, \eqref{mdef}, i.e. 
\be
m_n(\mathbb{I}|\alpha) := (-1)^{w(\mathbb{I}|\alpha)}\sum_{\Gamma\in \Omega(\mathbb{I})\bigcap\Omega(\alpha)}R_\Gamma(S).
\ee
It is known that there are $\texttt{C}_{n-2}$ Feynman diagrams which are planar with respect to a given order, where the $\texttt{C}_m$ are the standard Catalan numbers. One way to see this is that there is a bijection between planar cubic Feynman diagrams and triangulations of an $n$-gon. These can also be thought of as the vertices of an associahedron. Finding the diagrams that are shared by two orderings is equivalent to finding the intersection of two associahedra. The set of all such intersections with the canonical order associahedron corresponds to all possible subdivisions of an $n$-gon.

Luckily, the number of all such subdivisions is also well-known and it is given by the super Catalan or Schröder–Hipparchus numbers $\texttt{S}_{n}$. The first few corresponding to $n=4,5,6,7,8$ are $3, 11, 45, 197, 903, 4279$ respectively (see e.g. the sequence A001003 in the OEIS \cite{OEIS}).

Having found the number of cycles that give a non-trivial intersection with the canonical order, the complement, i.e., the number of orthogonal cycles is then given by 
\be
|\ort (\mathbb{I}_n)|= \frac{(n-1)!}{2} - \texttt{S}_{n}.
\ee
The number $|\ort (\mathbb{I}_n)|$ gives a lower bound on the number of cycles that can be chosen without succeeding to construct a basis of $\mathbb{R}^{(n-3)!}$. Let us see how this compares to $(n-2)!/2$, our lower bound on the number of compatible cycles when $G$ only has even cycles.

Let us consider the asymptotic behavior of the super Catalan numbers,
\be
\log(\texttt{S}_{n})\sim n\, \log\left(3+\sqrt{8}\right)-\frac{3}{2}\log(n) + {\cal O}(n^0). 
\ee

This number is very small compared to the total number of cycles $\log((n-1)!/2)\sim n\,\log (n)-n-1/2\log(n)+{\cal O}(n^0).$ 

Clearly, $|\ort (\mathbb{I}_n)| = |\ort (\alpha)|$ for any ordering $\alpha$ as can be seen from the definition of the map $\phi$ reviewed in section 2 and applied to cycles in \eqref{defPhi}. This means that the $(n-1)!/2\times (n-1)!/2$ matrix $m_n(\alpha|\beta)$ is very sparse when $n$ is large and it is increasingly difficult to find a basis of the space. 

This sparsity is even stronger than that expected from the block diagonal shape of the so-called KLT kernel \cite{Bern:1998ug}. Consider a basis for $\alpha$ of the form $(1,\omega ,n-1,n)$ and one for $\beta$ of the form $(1,\gamma_1, n,\gamma_2,n-1)$ with $|\gamma_1|-|\gamma_2|=(n+1\!\mod 2)$. In this case the $(n-3)!\times (n-3)!$ matrix ${\cal S}^{\rm KLT}_{\alpha,\beta}:=m_n^{-1}(\alpha|\beta)$ is known to be block diagonal with blocks of size $d\times d$ with $d=([n/2]-1)!([n/2]-2)!$ if $n$ is even and $d=(([n/2]-2)!([n/2]-2)!)$ if $n$ is odd. The blocks are completely solid, i.e., they do not possess any vanishing entries. Of course, the $(n-3)!\times (n-3)!$ matrix $m_n(\alpha|\beta)$ is also block diagonal. However, somewhat unexpectedly each block becomes sparse already for $n\sim 40$. Moreover, the sparsity increases as $n$ does since the ratio of $\texttt{S}_{n}$ to the size of a single block tends to zero as $n$ grows.  

The behavior of the linear relations among the vectors $\phi(C_\alpha)$ as $n$ grows can have important consequences not only for the construction considered in this work but also for the KLT procedure which connects theories such as Yang-Mills and gravity. The study of linear dependencies is an area in mathematics known as matriod theory \cite{ox}. The collection of all vectors $\phi(C_\alpha)$ defines a matroid of rank $(n-3)!$ on a ground set of $(n-1)!/2$ elements. For $n=4,5$ one has what is known as the uniform matroids $U_{1,3}$ and $U_{2,12}$ respectively. For $n>5$ the matroids have much more structure. For example, for $n=6$ we have found that there are $126,820$ bases for the submatroid defined by the $24$ orderings $\alpha =(1,2,\omega)$ with $\omega$ a permutation of $\{3,4,5,6\}$ which form what is known as the Kleiss-Kuijf set of orderings \cite{Kleiss:1988ne}.  

We leave a more in depth study of linear independence, asymptotic structure of the matrix $m_n(\alpha,\beta)$ and the matroids defined by the map $\phi$ to future work. 

\subsection{Outlook}

According to the numerical data gathered in \cite{Cachazo:2015nwa}, when the number of vertices $n$ is fixed, graphs with the largest number of cycles always have the smallest number of compatible cycles. When $n$ is even, such graphs are those with $n/2$ cycles and the number of compatible cycles was determined in section 4 from the connection to breakpoint graphs. If the behavior found in \cite{Cachazo:2015nwa} is correct, then it is clear that studying 2-regular graphs with only two cycles should be a natural starting point for the construction of a basis of compatible cycles, i.e., a set of $(n-3)!$ linearly independent vectors.  

This observation suggests a natural generalization to the proposal of \cite{Cachazo:2015nwa} for computing $m_n(G_1|G_2)$ in which the procedure is carried out in steps determined by the number of cycles in $G_i$. 

Start with the set of all 2-regular graphs with only two cycles ${\cal G}^{\hbox{{\small 2-reg}}}_{\rm 2\, cycles}$ and then compute all possible vectors $\phi(G)$ with $G\in {\cal G}^{\hbox{{\small 2-reg}}}_{\rm 2\, cycles}$ as a linear combination of vectors $\phi(C_\alpha)$ using their basis of compatible cycles, assuming it exists. Once this is done one can extend the set of compatible cycles to include ``compatible graphs" with $2$ cycles. 

\begin{defn}
Given a 2-regular graph $G$, a \emph{compatible graph} to $G$ is a 2-regular graph $B$ with a single or two cycles such that the 4-regular graph obtained by edge-disjoint union $G\cup B$ admits a decomposition of the form $G\cup B = B_1\cup B_2$ where $B_1$ and $B_2$ are both graphs with a single or two cycles.
\end{defn}

This means that the main problem can also be modified accordingly. 

\begin{problem}\label{main problem version two}
Given a 2-regular graph $G$ on $n$ vertices, find at least $(n-3)!$ compatible graphs such that under $\phi$ they form a basis of $\mathbb{R}^{(n-3)!}$.
\end{problem}

Clearly the set of compatible graphs to a given 2-regular graph is larger than the number of compatible cycles. Therefore, even if finding a set of $(n-3)!$ linearly independent vectors gets harder as $n$ increases, as suggested by the discussion above, one can compensate by enlarging the set to compatible graphs.  

This notion can be further extended to recursively include graphs with three, four cycles, etc. It would be very interesting to explore this further and the connection of this more general notion of compatibility with breakpoint graphs with more cycles.

\section*{Acknowledgements}
We would like to thank A. Guevara and S. Mizera for useful discussions and J. Bourjaily for bringing \cite{Bjerrum-Bohr:2016axv} to our attention. This research was supported in part by Perimeter Institute for Theoretical Physics. Research at Perimeter Institute is supported by the Government of Canada through the Department of Innovation, Science and Economic Development Canada and by the Province of Ontario through the Ministry of Research, Innovation and Science. KY is supported by an NSERC Discovery grant and was supported during this research by a Humboldt fellowship from the Alexander von Humboldt foundation.

\renewcommand{\thefigure}{\thesection.\arabic{figure}}
\renewcommand{\thetable}{\thesection.\arabic{table}}
\appendix


\bibliographystyle{JHEP}
\bibliography{references}

\end{document}